\documentclass[runningheads]{llncs}
\usepackage{graphicx}
\usepackage{amsmath}
\usepackage{amssymb}
\usepackage{hyperref}
\usepackage{xcolor}
\usepackage[utf8]{inputenc}
\usepackage[T1]{fontenc}
\usepackage{tikz}
\usepackage{parskip}
\usepackage[linesnumbered,ruled,vlined]{algorithm2e}
\usepackage{mathtools}
\DeclarePairedDelimiter{\ceil}{\lceil}{\rceil}
\DeclarePairedDelimiter{\floor}{\lfloor}{\rfloor}
\spnewtheorem{obs}[theorem]{Observation}{\bfseries}{\rmfamily}

\begin{document}
\title{The Connected Domination Number of Grids
\thanks{This research was supported by the first author's INSPIRE fellowship from Department of Science and Technology (DST), Govt. of India}}
%
\author{Adarsh Srinivasan\inst{1}\orcidID{0000-0003-0288-4818} \and N S Narayanaswamy\inst{2}\orcidID{0000-0002-8771-3921} }
\authorrunning{Adarsh Srinivasan \and N S Narayanaswamy}
%
\institute{Indian Institute of Science Education and Research (IISER) Pune \and
Indian Institute of Technology (IIT) Madras\\ 
\email{adarsh.srinivasan@students.iiserpune.ac.in}\\
\email{swamy@cse.iitm.ac.in}}
\maketitle          
\begin{abstract}
Closed form expressions for the domination number of an $n \times m$ grid have attracted significant attention, and an exact expression has been obtained in 2011~\cite{griddomset1}. In this paper, we present our results on obtaining new lower bounds on the connected domination number of an $n \times m$ grid. The problem has been solved for grids with up to $4$ rows and with $6$ rows and the best currently known lower bound for arbitrary $m,n$ is $\ceil*{\frac{mn}{3}}$~\cite{gridbounds}. Fujie~\cite{FUJIE20031931} came up with a general construction for a connected dominating set of an $n \times m$ grid. In this paper, we investigate whether this construction is indeed optimum. We prove a new lower bound of $\ceil*{\frac{mn+2\ceil*{\frac{\min \{m,n\}}{3}}}{3}}$ for arbitrary $m,n \geq 4$.

\end{abstract}
\section{Introduction}
In this paper, we study the \textsc{Minimum Connected Dominating Set} (\textsc{Min-CDS}) problem in grid graphs. Given a connected graph $G=(V,E)$, a \emph{connected dominating set} (CDS) is a subset $S$ of $V$ which induces a connected subgraph in $G$ such that every vertex of $G$ is either in $S$ or adjacent to a vertex in $S$. The \textsc{Min-CDS} problem asks for a CDS of minimum size. This is a well studied problem in combinatorial optimisation. The \emph{connected domination number} of $G$ is the minimum size of a connected dominating set of $G$. This problem is equivalent to the \textsc{Maximum Leaf Spanning Tree} (\textsc{MLST}) problem, which is the problem of finding a spanning tree of $G$ with maximum number of leaves. A graph has a spanning tree with $k$ leaves if and only if it has a connected dominating set of size $|V|-k$. These problems are known to be NP-complete~\cite[ND2, Appendix 2]{garey}, and have been widely studied and have applications in areas such as networking, circuit layout, etc (See \cite{uses1} for example). A common theme in the study of any NP-complete problem is to consider the problem in special classes of inputs with more structure than the general case and try to understand whether the problem remains NP-complete or admits a polynomial time solution. \textsc{Min-CDS} is known to be NP-complete when it is restricted to planar bipartite graphs of maximum degree $4$~\cite{bipartiteNP}. It is also known to be NP-complete for unit disk graphs~\cite{lichtenstein1982planar} and subgraphs of grid graphs~\cite[Theorem 6.1]{CLARK1990165}. When viewed in terms of approximation algorithms, the minimum connected dominating set and maximum leaf spanning tree problems are not equivalent. The MLST problem is MAX-SNP-hard which makes a Polynomial Time Approximation Scheme (PTAS) unlikely~\cite{MLSTApprox}, but linear time $3$-approximation algorithms~\cite{lu1998approximating} and $2$-approximation algorithms~\cite{solis20172} exist. A PTAS exists for the \textsc{Min-CDS} problem on unit disk graphs~\cite{du2012polynomial,hunt1998nc}. The complexity of this problem for complete grid graphs remains unknown.

By comparison, the computation of the domination number of an $n \times m$ grid graph is a well studied problem. Chang~\cite{chang} devoted their PhD thesis to calculating the domination number of grids and Gonçalves et al~\cite{griddomset1} solved this problem in 2011 by proving it to be $\floor*{\frac{(n+2)(m+2)}{5}}-4$ for $16 \leq m \leq n$. Hence it is natural to ask the following questions about the connected domination number of grid graphs:
\begin{itemize}
    \item Can we come up with a closed form expression for the connected domination number of an $n \times m$ grid?
    \item Can we design an algorithm that takes $n,m$ as input, with run-time polynomial in $n,m$ that outputs the domination number of an $n \times m$ grid?
\end{itemize}
An answer for the first question would imply an answer for the second one, but not vice versa. A partial answer can be obtained by showing lower and upper bounds on the connected domination number. Upper bounds can be obtained using constructions or heuristic algorithms. Fujie~\cite{FUJIE20031931} came up with a general construction of a spanning tree of a grid with a large number of leaves which leads to an upper bound on the connected domination number. Li and Tolouse~\cite{gridbounds} determined the optimum maximum leaf spanning tree for grid graphs with up to 4 rows and with 6 rows. The only known general lower bound is $\ceil*{\frac{mn}{3}}$. This was obtained by Li and Tolouse using an easy counting argument and by Fujie using a mathematical programming approach. 

In this paper, we come up with improved lower bounds on the connected domination number of a grid. We show a lower bound of $\ceil*{\frac{mn+2\ceil*{\frac{\min \{m,n\}}{3}}}{3}}$ for arbitrary $m,n \geq 4$ (Theorem \ref{mainthm}). To our knowledge, this is the first non-trivial result of this kind. Our proof also leads to some insight on the structure of the optimum connected dominating set of a grid. 
\subsection{Preliminaries and Terminology}
We first introduce the definitions and notations that we will use in the rest of the paper. For all definitions and notations not defined here, we refer to~\cite{west}. Let $G=(V,E)$ be a connected graph. A leaf refers to a vertex of degree $1$ in $G$. The open neighbourhood of a subset $S$ of $V$ in $G$ is defined to be the set of all vertices adjacent to a vertex in $S$ which are not in $S$ and denoted by $N_G(S)$. The closed neighbourhood of $S$ in $G$, $N_G[S]$ is defined to be $N_G(S) \cup S$. $G[S]$ denotes the subgraph of $G$ induced by $S$. A set $S$ is called a connected dominating set of $G$ if $N_G[S]=V$ and $G[S]$ is a connected subgraph of $G$. The size of the minimum connected dominating set of $G$ is called its connected domination number and is denoted by $\gamma_c(G)$. The maximum leaf number of $G$ is the number of leaves in the maximum leaf spanning tree of $G$. The connected domination number and the maximum leaf number add up to $|V|$. A connected dominating set of $G$ can be obtained by deleting the leaves of a spanning tree of $G$.

The notation $[i]$ denotes the set $\{1,2,\dots,i\}$. The $n \times m$ grid graph $G_{n,m}$ is the graph with the vertex set $[n] \times [m]$ with two vertices $(i_1,j_1)$ and $(i_2,j_2)$ being adjacent if and only if $|i_i-i_2|+|j_1-j_2|=1$. It can also be defined as a unit disk graph in which the disks have the integer points mentioned as centers and radius $1/2$. For the reminder of the paper, we just use $G$ instead of $G_{n,m}$ without any ambiguity. We assume whenever necessary that $G$ is embedded in a larger grid graph. Specifically, we embed $G$ in $G'$ which is a grid graph with two additional rows and columns. The vertex set of $G'$ is $\{0,1,\dots,n+1\} \times \{0,1,\dots,m+1\}$ with the same incidence relation as $G$.

$S$ is a connected dominating set of $G$. $l$ is the number of leaves in the graph $G[S]$. It has no relation to the number of leaves in the corresponding spanning tree of $G$. For any $v \in G$ we define the \emph{loss function} of that vertex to be $\ell(v)=|N[v] \cap S|-1$. The loss function of the set $S$ is defined to be $\ell(S)=\sum_{v \in N_{G'}[S]}\ell(v)$. The \emph{boundary} of $G_{n,m}$ is defined to be the set of points in $G_{n,m}$ which have three neighbours or less in $G_{n,m}$ (excluding the points themselves). The \emph{excess function} $e(S)$ is defined to be the number of points in $S$ present in the boundary of $G$. These definitions are inspired by similar definitions in~\cite{griddomset1}.
\section{Bounds on the connected domination number}
\subsection{Known upper bounds}\label{const}
Upper bounds for $\gamma_c(G)$ can be easily obtained by constructing spanning trees for $G$ with a large number of leaves, which leads to an upper bound for the maximum leaf number of $G$, and a corresponding lower bound on the connected domination number. Fujie gave a construction of a spanning tree with a large number of leaves~\cite[Lemma 2]{FUJIE20031931}. We reproduce their construction here:

Let $D_1$ be a CDS of $G_{n,m}$ with the following vertices- 
\begin{align*}
    &(1,2),(2,2),\dots,(n,2)\\
    &(1,m-1),(2,m-1),\dots (n,m-1)\\
    &(2,3),(2,4),\dots,(2,m-2)\\
    &(i,3k+2) \text{ for } i=3,4,\dots,n,k=1,2,\dots,\floor*{\frac{m-4}{3}}
\end{align*}
Let $D_2$ be a CDS with the following vertices-
\begin{align*}
    &(2,1),(2,2),\dots,(2,m)\\
    &(n-1,1),(n-1,2),\dots (n-1,m)\\
    &(3,2),(4,2),\dots,(m-2,2)\\
    &(3k+2,i) \text{ for } k=1,2,\dots,\floor*{\frac{n-4}{3}},i=3,4,\dots,m
\end{align*}
Hence, we have the following upper bound on $\gamma_c(G)$:
\begin{equation} \label{fujie}
    \gamma_c(G) \leq \min \left\{2n+(m-4)+\floor*{\frac{m-4}{3}}(n-2), 2m+(n-4)+\floor*{\frac{n-4}{3}}(m-2) \right\} 
\end{equation} 
Fig. \ref{cons} describes an example for the constructions for the graph $G_{7,11}$.
\begin{figure}
    \centering
    \begin{tikzpicture}[scale=0.40]
    \foreach \i in {0,...,10}
    \foreach \j in {0,...,6}{
    \draw (\i,\j) circle(3pt); }
    \draw[gray] (0,0) grid (10,6);
    \foreach \i in {0,...,10}
    \foreach \j in {0,...,6}{
    \fill[white] (\i,\j) circle(3pt); }
    \foreach \i in {0,...,6}{
    \fill[black] (1,\i) circle(3pt);  }
    \foreach \i in {0,...,6}{
    \fill[black] (9,\i) circle(3pt);  }
    \foreach \i in {2,...,8}{
    \fill[black] (\i,1) circle(3pt);  }
    \foreach \i in {2,...,6}{
    \fill[black] (4,\i) circle(3pt);  }
     \foreach \i in {2,...,6}{
    \fill[black] (7,\i) circle(3pt);  }
    \end{tikzpicture}
    \qquad
    \qquad
    \begin{tikzpicture}[scale=0.40]
    \foreach \i in {0,...,10}
    \foreach \j in {0,...,6}{
    \draw (\i,\j) circle(3pt); }
    \draw[gray] (0,0) grid (10,6);
    \foreach \i in {0,...,10}
    \foreach \j in {0,...,6}{
    \fill[white] (\i,\j) circle(3pt); }
    \foreach \i in {0,...,10}{
    \fill[black](\i,1) circle(3pt);}
    \foreach \i in {0,...,10}{
    \fill[black](\i,5) circle(3pt);}
    \foreach \i in {2,...,10}{
    \fill[black](\i,4) circle(3pt);}
    \foreach \i in {2,...,4}{
    \fill[black](1,\i) circle(3pt);}
    \end{tikzpicture}
    \caption{$D_1$ and $D_2$ for $G_{7,11}$ (Black vertices present in CDS)}
    \label{cons}
\end{figure}
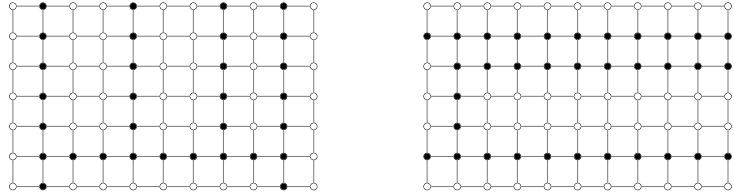
\subsection{New lower bounds}
Proving lower bounds would require combinatorial arguments, which will be the main contribution of this paper. In \cite{griddomset1} the authors introduced a combinatorial parameter called the loss function to prove lower bounds on the domination number of a grid. The loss and excess functions which we have defined are inspired by that definition. Using these parameters, we prove a sequence of lower bounds on $|S|$, each an improvement on the previous one, culminating in theorem \ref{mainthm}.

Our approach is to obtain lower bounds in $\ell(S)$ and $e(S)$ parametrized by $l$ and then combine them to obtain absolute lower bounds on $\ell(S)+e(S)$. This will, in turn lead to lower bounds on $|S|$.
	
$G[S]$ can be divided into a number of horizontal and vertical line segments, with each vertical line segment connected to at least one horizontal line segment and vice versa as the graph is connected. The vertices, called \emph{joins} where a horizontal line segment meets a vertical line segment can be of degree $2$, $3$ or $4$.  When the join is of degree $2$, we refer to the horizontal line segment and vertical line segment that meet at the vertex as a \emph{bend}.
	
Let $d_3$ and $d_4$ be the number of vertices in $G[S]$ with degree $3$ and $4$ respectively and $d_2$ be the number of bends. Note that here $d_2$ counts only those degree $2$ vertices which form a bend. We make the following observation on the number of leaves in $G[S]$:
\begin{lemma} \label{handshake}
    For any CDS $S$ of $G$: $l \leq d_3+2d_4+2 $
\end{lemma}
\begin{proof}
    The well known handshake lemma states that for a graph $G=(V,E)$ $\sum_{v \in V} \delta(v)=2|E|$. We apply this for $G[S]$. The number of vertices of degree $2$ is $|S|-d_3-d_4-l$ and as $G[S]$ is connected, the number of edges in $G[S]$ is at least $|S|-1$.
    \begin{align*}
        3d_3+4d_4+2(|S|-d_3-d_4-l)+l=2|E| \geq 2(|S|-1)&\\
        \implies l \leq d_3+2d_4+2&
    \end{align*}
    \qed
\end{proof}
We now relate the parameters $\ell(S)$ and $e(S)$ to the size of a connected dominating set:
\begin{lemma}
For any $m,n\geq 3$, $G_{n,m}$ has a minimum CDS that does not contain any corner of $G_{n,m}$ 
\end{lemma}
\begin{proof}
Consider a CDS $S$ of $G_{n,m}$ which contains the corner point $(1,1)$. As $G[S]$ is a connected subgraph, $S$ must contain either $(1,2)$ or $(2,1)$. There exists a maximal horizontal or vertical line segment in $G[S]$ containing $(1,1)$. Assume $G[S]$ contains the path $(1,1),(1,2),\dots, (1,k)$ As $G[S]$ is connected, one of these points must contain a neighbour in $S$. Let $(1,i)$ be the first such point with a neighbour $(2,i)$. Now, for all $j < i$, we replace $(1,j)$ in $S$ with $(2,j)$ to obtain a new CDS of $G$ with the same number of points. As $n,m \geq 4$, $(2,1)$ which replaces $(1,1)$ is not a corner point. We perform a similar procedure with the path $(1,1),(2,1),\dots,(k,1)$ if $S$ does not contain the point $(1,2)$, and repeat this for all four corner points of $G$ to obtain a new CDS for $G_{n,m}$ with the same number of points as $S$. 
\qed
\end{proof}
\begin{lemma} \label{bnd}
    For a CDS $S$ of $G$: $nm = 5|S|-\ell(S)-e(S)$
\end{lemma}
\begin{proof}
For a set $S$ which dominates $G$, consider the set $N_{G'}[S]$, which is its closed neighbourhood in $G'$. Any point in $S$ dominates $5$ points including itself, and for each $v \in N_{G'}[S]$, the number of points which dominate $v$ is $1+\ell(v)$. Hence, $|N_{G'}[S]|=5|S|-\ell(S)$.

As no point in $S$ is a corner point of $G$ (see previous lemma), every point in the boundary of $G$ dominates exactly one point in $G'$ outside $G$. Thus the number of points in $N_{G'}[S]$ outside $G$ is $e(S)$. Hence, $mn=|N_{G'}[S]|-e(S)$ and this proves the lemma.\qed
\end{proof}
\begin{figure}[ht] 
	\centering
	\begin{tikzpicture}
	\foreach \i in {0,...,2}
	\foreach \j in {0,...,2}{
		\draw (\i,\j) circle(3pt);}
	\fill[black] (1,1) circle(3pt);
	\fill[black] (2,1) circle(3pt);
	\fill[black] (0,1) circle(3pt);
	\fill[black] (1,0) circle(3pt);
	\fill[black] (1,0) circle(3pt);
	\draw (1,1) circle(3pt) node[black, above left] {\textbf{P}};
	\draw (0,0) circle(3pt) node[black, above right] {\textbf{Q}};
	\draw (2,0) circle(3pt) node[black, above left] {\textbf{R}};
	\draw (0,1)--(1,1);
	\draw (1,0)--(1,1);
	\draw (1,1)--(2,1);
	\end{tikzpicture}
	\quad
	\quad
	\begin{tikzpicture}
	\foreach \i in {0,...,2}
	\foreach \j in {0,...,2}{
		\draw (\i,\j) circle(3pt);}
	\fill[black] (1,1) circle(3pt);
	\fill[black] (2,1) circle(3pt);
	\fill[black] (0,1) circle(3pt);
	\fill[black] (1,0) circle(3pt);
	\fill[black] (1,2) circle(3pt);
	\draw (1,1) circle(3pt) node[black, above left] {\textbf{P}};
	\draw (0,0) circle(3pt) node[black, above right] {\textbf{Q}};
	\draw (2,0) circle(3pt) node[black, above left] {\textbf{R}};
	\draw (0,2) circle(3pt) node[black, below right] {\textbf{S}};
	\draw (2,2) circle(3pt) node[black, below left] {\textbf{T}};
	\draw (0,1)--(1,1);
	\draw (1,0)--(1,1);
	\draw (1,1)--(2,1);
	\draw (1,1)--(1,2);
	\end{tikzpicture}
	\quad 
	\quad
	\begin{tikzpicture}
    \foreach \i in {0,...,2}
    \foreach \j in {0,...,2}{
    \draw (\i,\j) circle(3pt);}
    \fill[black] (1,1) circle(3pt);
    \fill[black] (1,0) circle(3pt);
    \fill[black] (0,2) circle(3pt);
    \fill[black] (1,2) circle(3pt);
    \draw (0,1) circle(3pt) node[black, above right] {\textbf{P}};
    \draw (0,2)--(1,2);
    \draw (1,2)--(1,1);
    \draw (1,1)--(1,0);
    \end{tikzpicture}
	\caption{The three different types of joins}
	\label{joins}
\end{figure}
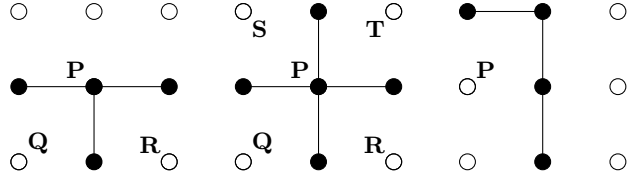
\begin{lemma} \label{main}
    $\ell(S) \geq 2|S|-l+d_2+3d_3+6d_4$ and $e(S) \geq 4$ if $m,n \geq 4$ 
\end{lemma}
\begin{proof}
Consider the four corners of the grid, the points $(1,1), (n,1), (1,m)$ and $(n,m)$. These points have to be dominated by a point in $S$, and all their neighbours in $G$ are in the boundary of $G$. As both $n$ and $m$ are greater than or equal to $4$, two corner points cannot be dominated by the same point in $S$, and hence $e(S) \geq 4$. 

For any point $v$ in $S$ which is not a leaf in $G[S]$, $\ell(v) \geq 2$. This is because it has at least $2$ neighbours in $S$. If $v$ is a leaf, $\ell(v)=1$. Hence, if $l$ is the number of leaves of $S$, $\ell(S) \geq 2|S|-l$. In addition to this, consider a vertex of degree $3$ as in Fig. \ref{joins}. The loss function of the point $P$ is at least $3$ and the loss function of $Q$ and $R$ is at least $1$. For the vertex of degree $4$ shown, the loss function of $P$ is at least $4$ and the loss functions of $Q,R,S$ and $T$ are each at least $1$. For the bend shown, the loss function of the point $P$ is at least $1$. Hence, $\ell(S) \geq 2|S|-l+d_2+3d_3+6d_4$. \qed
\end{proof}
Putting these observations together, we get our first bound on $e(S)+\ell(S)$:
\begin{lemma}[Parametrized bound 1]
	Consider any CDS $S$ for $G$, with $G[S]$ having $l$ leaves. Then $\ell(S)+e(S) \geq 2|S|+2l-2$.
\end{lemma}
\begin{proof}
This follows from the fact that $\ell(S)+e(S) \geq 2|S|-l+3d_3+6d_4+4$ (lemma \ref{main})  and $d_3+2d_4 \geq l-2$ (lemma \ref{handshake}). \qed
\end{proof}
From this bound it is easy to derive the already known lower bound of $\ceil*{\frac{mn}{3}}$ for $\gamma_c(G)$:
\begin{theorem}[Bound 1]
    For any CDS $S$ of $G$:
    \begin{equation*}
        |S| \geq \ceil*{\frac{mn}{3}}
    \end{equation*}
\end{theorem}
\begin{proof}
As $S$ must have at least $2$ leaves, $\ell(S)+e(S) \geq 2|S|+2$. We can now use this in lemma \ref{bnd}:
\begin{align*}
    mn \leq 5|S|-2|S|-2=3|S|-2\\
    \implies |S| \geq \ceil*{\frac{mn}{3}}
\end{align*}\qed
\end{proof}
We have not used any structural information on $G[S]$ yet. Specifically, we have not used the fact that it might contain bends. Next, we use the fact that the connected dominating set must contain a certain minimum number of joins to show a new bound on $\ell(S)+e(S)$. We first prove the following simple lemma on the structure of $G[S]$. We say that the horizontal line segments \emph{span} the height of the graph if the subgraph induced by their closed neighbourhood contains a point from every row of $G$ and we say that the vertical line segments \emph{span} the width of the graph if the subgraph induced by their closed neighbourhood contains a point from every column of $G$.
\begin{lemma}
	$G[S]$ either has at least $\ceil*{\frac{n}{3}}$ horizontal line segments which span the width of $G$ or $\ceil*{\frac{m}{3}}$ vertical line segments which span the height of $G$.
\end{lemma}
\begin{proof}
	Any horizontal line segment dominates an area that spans at most three rows. Hence, if $S$ has less than $\ceil*{\frac{n}{3}}$ horizontal line segments, there exists at least one row which is not dominated by any of the points in the horizontal line segments. The rows not dominated by horizontal line segments must be dominated by the vertical line segments. Any vertical line segment can dominate an area which spans at most $3$ columns and hence there must be at least $\ceil*{\frac{m}{3}}$ vertical line segments in $S$. Similarly, if $S$ has less than $\ceil*{\frac{m}{3}}$ vertical line segments, it must have at least $\ceil*{\frac{n}{3}}$ horizontal line segments. \qed
\end{proof}
We use this lemma to prove the following bound on the number of joins in $G[S]$.
\begin{lemma}\label{joins1}
Let $d_2$, $d_3$ and $d_4$ denote the number of bends, joins of degree $3$ and joins of degree $4$ respectively.
\begin{equation*}
    d_2+d_3+d_4 \geq \ceil*{\frac{\min \{m,n\}}{3}}
\end{equation*}
\end{lemma}
\begin{proof}
Either the horizontal line segments dominate an area that spans the entire height of the graph or the vertical line segments dominate an area that spans the width of the graph. We assume the latter without loss of generality. Each of these vertical line segments must be connected to a point in the previous or next column in the grid. Hence, each of these must contain either a bend, a degree $3$ join or a degree $4$ join. From the previous lemma, we know that there are at least $\ceil*{\frac{m}{3}}$ vertical line segments and the result follows. \qed
\end{proof}
We use these two lemmas to obtain two new lower bounds on $\ell(S)+e(S)$ parametrized by the number of leaves, which we can then combine to obtain an improved lower bound for $|S|$.
\begin{lemma}[parametrized bound 2]\label{pb2}
	Consider any CDS $S$ for $G$, with $G[S]$ having $l$ leaves. Then
	$\ell(S)+e(S) \geq 2|S|+\ceil*{\frac{\min \{m,n\}}{3}}+l$
\end{lemma}
\begin{proof}	
We know that $e(S) \geq 4$ and $\ell(S) \geq 2|S|-l+3d_3+6d_4+d_2$, and because $d_2 +d_3+d_4 \geq \ceil*{\frac{\min \{m,n\}}{3}}$, 
\begin{equation*}
    \ell(S) \geq 2|S|-l+\ceil*{\frac{\min \{m,n\}}{3}}+2d_3+5d_4 \geq 2|S|-l+\ceil*{\frac{\min \{m,n\}}{3}}+2(d_3+2d_4+2)-4
\end{equation*}
We now use the fact that $l \leq d_3+2d_4+2$ to prove the lemma \qed
\end{proof}
\begin{lemma}[parametrized bound 3] \label{pb3}
	Consider any CDS $S$ for $G$, with $G[S]$ having $l$ leaves. Then
	$\ell(S)+e(S) \geq 2|S|+2\ceil*{\frac{\min \{m,n\}}{3}}+2-l$
\end{lemma}
\begin{proof}
	We can assume without loss of generality that the vertical lines span the width of the grid. Every vertical line must contain a join. Hence, at least $\ceil*{\frac{m}{3}}-d_3-d_4$ of them must have one or more bends. A vertical line with only one bend must also contain a leaf. This implies that the number of bends is at least $2\left(\ceil*{\frac{\min \{m,n\}}{3}}-d_3-d_4\right)-l$. We use this in our estimation of $\ell(S)$: 
	\begin{align*}
	    \ell(S) \geq & 2|S|-l+d_2+3d_3+6d_4 \\
	    \geq & 2|S|-l+3d_3+6d_4+2\left(\ceil*{\frac{\min \{m,n\}}{3}}-d_3-d_4\right)-l  \\
	    \geq & 2|S|+2\ceil*{\frac{\min \{m,n\}}{3}} +d_3+2d_4+2-2l-2\\
	    \geq & 2|S| +2\ceil*{\frac{\min \{m,n\}}{3}}-l-2 
	\end{align*}
	Using the fact that $e(S) \geq 4$, the result follows.\qed
\end{proof}
Combining the previous two parametrized bounds leads to the following lower bound on $|S|$ which is an improvement over the currently known bound of $|S| \geq \ceil*{\frac{mn}{3}}$:
\begin{theorem}[Bound 2]
	For a CDS $S$ of $G$:
	\begin{equation*}
	    |S| \geq \ceil*{\frac{mn+\ceil*{\frac{3}{2}\ceil*{\frac{\min \{m,n\}}{3}}}+1}{3}}
	\end{equation*}
\end{theorem} 
\begin{proof}
The lower bound in lemma \ref{pb2} increases with $l$ and the bound in lemma \ref{pb3} decreases with $l$. As they both lower bound $\ell(S)+e(S)$, $\ell(S)+e(S)$ is always greater than or equal to $2|S|+\ceil*{\frac{3}{2}\ceil*{\frac{\min \{m,n\}}{3}}}+1$. 
From Lemma \ref{bnd},
\begin{equation*}
    mn = 5|S|-\ell(S)-e(S) \geq 5|S| -\left( 2|S|+\ceil*{\frac{3}{2}\ceil*{\frac{\min \{m,n\}}{3}}}+1 \right)
\end{equation*}
This means that $3|S| \geq mn+\ceil*{\frac{3}{2}\ceil*{\frac{\min \{m,n\}}{3}}}+1$ which proves the theorem.\qed 
\end{proof}
This bound can be further improved by counting the number of bends in $G[S]$ more carefully. In the proof of lemma \ref{pb3}, we used the fact that the number of bends is at least $2\left(\ceil*{\frac{\min \{m,n\}}{3}}-d_3-d_4\right)-l$. In the following lemma, we improve on that:
\begin{lemma}
Consider a CDS $S$ of $G$. The number of bends in $S$ is at least $2\left( \ceil*{\frac{\min \{m,n \}}{3}}-l+1\right)$. 
\end{lemma}
\begin{proof}
We first assume that $G[S]$ has no vertices of degree $4$. If it does, we can just treat a vertex of degree $4$ as two vertices of degree $3$. Hence, the number of vertices of degree $3$ has to be at least $l-2$. As before, we can assume without loss of generality that the vertical lines dominate an area that spans the width of the grid and that there are $\ceil*{\frac{m}{3}}$ vertical lines. Some of these lines have one or more degree $3$ vertices and some of them have bends. 

Observe that a vertical line with only one join must contain a leaf. Every join in a vertical line must be paired with another join or be paired with a leaf as shown in Fig.~\ref{deg3}.

\begin{figure}[ht]
    \centering
    \begin{tikzpicture}
    \foreach \i in {0,...,2}
    \foreach \j in {0,...,2}{
    	\draw (\i,\j) circle(3pt);}
    \fill[black] (1,2) circle(3pt);
    \fill[black] (1,1) circle(3pt);
    \fill[black] (1,0) circle(3pt);
    \fill[black] (2,2) circle(3pt);
    \fill[black] (0,2) circle(3pt);
    \fill[black] (0,0) circle(3pt);
    \draw (0,2)--(1,2);
    \draw (1,2)--(1,1);
    \draw (1,1)--(1,0);
    \draw (1,2)--(2,2);
    \draw (0,0)--(1,0);
    \end{tikzpicture}
    \qquad
    \begin{tikzpicture}
    \foreach \i in {0,...,2}
    \foreach \j in {0,...,2}{
    	\draw (\i,\j) circle(3pt);}
    \fill[black] (1,2) circle(3pt);
    \fill[black] (1,1) circle(3pt);
    \fill[black] (1,0) circle(3pt);
    \fill[black] (2,2) circle(3pt);
    \fill[black] (0,2) circle(3pt);
    \fill[black] (0,0) circle(3pt);
    \fill[black] (2,0) circle(3pt);
    \draw (0,2)--(1,2);
    \draw (1,2)--(1,1);
    \draw (1,1)--(1,0);
    \draw (1,2)--(2,2);
    \draw (0,0)--(1,0);
    \draw (1,0)--(2,0);
    \end{tikzpicture}
    \qquad
    \begin{tikzpicture}
    \foreach \i in {0,...,2}
    \foreach \j in {0,...,2}{
    	\draw (\i,\j) circle(3pt);}
    \fill[black] (1,2) circle(3pt);
    \fill[black] (1,1) circle(3pt);
    \fill[black] (1,0) circle(3pt);
    \fill[black] (2,2) circle(3pt);
    \fill[black] (0,2) circle(3pt);
    \draw (0,2)--(1,2);
    \draw (1,2)--(1,1);
    \draw (1,1)--(1,0);
    \draw (1,2)--(2,2);
    \end{tikzpicture}
    \qquad
    \begin{tikzpicture}
    \foreach \i in {0,...,2}
    \foreach \j in {0,...,2}{
    	\draw (\i,\j) circle(3pt);}
    \fill[black] (1,2) circle(3pt);
    \fill[black] (1,1) circle(3pt);
    \fill[black] (1,0) circle(3pt);
    \fill[black] (0,2) circle(3pt);
    \draw (0,2)--(1,2);
    \draw (1,2)--(1,1);
    \draw (1,1)--(1,0);
    \end{tikzpicture}
    \caption{}
    \label{deg3}
\end{figure}
Out of all the vertices of degree $3$, let $t_1$ be the number of vertices not paired with a bend or a degree $3$ vertex, $t_2$ be the number of bend-degree $3$ vertex pairs, and $t_3$ be the number of degree $3$ vertex-  degree $3$ vertex pairs. Hence, it is clear that out of the vertical lines, at most $(l-2)-t_3$ contain vertices of degree $3$. Hence there are at least $\left( \ceil*{\frac{m}{3}}-(l-2)+t_3\right)$ vertical lines without vertices of degree $3$ and these vertical lines contain at most $l-t_1$ leaves. Consider such a vertical line. It can have only one bend if and only if it has a leaf and hence there can be at most $(l-t_1)$ of such columns. We have already counted $t_2$ bends. Hence, we can bound the number of bends:
\begin{align*}
    d_2 &\geq 2\left( \ceil*{\frac{m}{3}}-(l-2)+t_3-(l-t_1)\right)+t_2+(l-t_1)\\
    &=2\ceil*{\frac{m}{3}}-3l+4+t_1+t_2+2t_3
\end{align*}
$t_1+t_2+2t_3 \geq l-2$ as every vertex of degree $3$ belongs in at least one of the three categories mentioned and the lemma follows. \qed
\end{proof}
We now have the necessary material to prove the main result of our paper.
\begin{theorem}[Main Theorem] \label{mainthm}
    For a CDS $S$ of $G$:
    \[ |S| \geq \ceil*{\frac{mn+2\ceil*{\frac{\min \{m,n\}}{3}}}{3}} \]
\end{theorem}
\begin{proof}
\begin{align*}
    \ell(S) \geq & 2|S|-l+d_2+3d_3+6d_4\\
    \geq & 2|S|-l+3d_3+6d_4+6+2\left(  \ceil*{\frac{\min \{m,n\}}{3}}-l+1\right)-6\\
    \geq & 2|S|+2\ceil*{\frac{\min \{m,n\}}{3}}-4
\end{align*}
We have used the fact that $d_3+2d_4+2\geq l$. As $e(S) \geq 4$, $\ell(S)+e(S) \geq 2|S|+2\ceil*{\frac{\min\{m,n\}}{3}}$, the theorem follows. \qed
\end{proof}
\subsection{Gap between lower and upper bounds}
In this section, we compare the gap between the lower and upper bounds obtained. To do that, we have to consider this case by case, for reminders $n$ and $m$ leave on division by $3$. We assume $m \leq n$. Let $L$ denote the lower bound obtained in theorem \ref{mainthm}. We let $L=\frac{mn}{3}+\frac{2m}{9}$, omitting the ceiling functions as they would only increase $L$ by at most $2$. Out of the two CDS's we constructed in section \ref{const}, $D_1$ and $D_2$, the upper bound is given by the construction of smaller size. If $m$ is divisible by $3$, $|D_1|\leq |D_2|$, and the gap between the lower and upper bounds is $\frac{m}{9}$. If $m$ is not divisible by $3$ and $n$ is however, then $|D_2| \leq |D_1|$ and the gap is $\frac{n}{3}-\frac{2m}{9}$. Similarly, we can analyse the other cases using Table \ref{tab1} and Table \ref{tab2}. The new lower bound is closest to the constructions in the case that $m$ is divisible by $3$.
\renewcommand{\arraystretch}{1.5}
\begin{table}[ht]
    \centering
    \begin{tabular}{|c|c|c|c|}
        \hline
         $m\mod{3}$ & $|D_1|$ &$|D_1|-\frac{mn}{3}$ & $|D_1|-L$   \\
         \hline
         $0$ & $\frac{mn}{3}+\frac{m}{3}$ & $\frac{m}{3}$ &  $\frac{m}{9}$\\
         \hline
         $1$ & $\frac{mn}{3}+\frac{m}{3}+\frac{2n}{3}-\frac{4}{3}$ & $\frac{m}{3}+\frac{2n}{3}-\frac{4}{3}$& $\frac{m}{9}+\frac{2n}{3}-\frac{4}{3} $\\
         \hline
         $2$ & $\frac{mn}{3}+\frac{m}{3}+\frac{n}{3}-\frac{2}{3}$ & $\frac{m}{3}+\frac{n}{3}-\frac{2}{3}$& $\frac{m}{9}+\frac{n}{3}-\frac{2}{3}$\\
         \hline
    \end{tabular}
    \caption{Gaps between $|D_1|$ and $L$}
    \label{tab1}
\end{table}
\begin{table}[ht]
    \centering
    \begin{tabular}{|c|c|c|c|}
        \hline
         $n \mod 3$ & $|D_2|$ &$|D_2|-\frac{mn}{3}$& $|D_2|-L$  \\
         \hline
         $0$ & $\frac{mn}{3}+\frac{n}{3}$ & $\frac{n}{3}$ & $\frac{n}{3}-\frac{2m}{9}$ \\
         \hline
         $1$ & $\frac{mn}{3}+\frac{n}{3}+\frac{2m}{3}-\frac{4}{3}$& $\frac{n}{3}+\frac{2m}{3}-\frac{4}{3}$& $\frac{n}{3}+\frac{4m}{9}-\frac{4}{3}$ \\
         \hline
         $2$ & $\frac{mn}{3}+\frac{m}{3}+\frac{n}{3}-\frac{2}{3}$ & $\frac{m}{3}+\frac{n}{3}-\frac{2}{3}$& $\frac{m}{9}+\frac{n}{3}-\frac{2}{3}$\\
         \hline
    \end{tabular}
    \caption{Gaps between $|D_2|$ and $L$}
    \label{tab2}
\end{table}
\section{Conclusions and further research}
In this paper, we come up with improved lower bounds on the connected domination number of a grid. The question of finding a closed form expression however, remains open. We have broadly used the following approach to prove lower bounds on $|S|$. Using the fact that $G$ is a grid graph, we obtained some structural results for any connected set that dominates $G$, which lead to lower bounds on the number of bends, vertices of degree $3$ and vertices of degree $4$ in $G[S]$. We then used lemma \ref{main} to get lower bounds on $|S|$. This approach however, does not capture the full picture. Consider Fujie's construction detailed in section \ref{const}. There are no bends or vertices of degree $4$ and the number of vertices of degree $3$ is $\ceil*{\frac{m}{3}}$ or $\ceil*{\frac{n}{3}}$, and our techniques have already accounted for this. There is still a gap between our lower bound and this upper bound because $e(S)=\ceil*{\frac{m}{3}}+2$ and $\ceil*{\frac{n}{3}}+2$ for these constructions, while we have used a lower bound of $4$ for $e(S)$. While this is the best possible lower bound for $e(S)$ separately, it might be possible to obtain better lower bounds for $\ell(S)+e(S)$ by trying to lower bound the sum of two quantities, rather than lower bound each quantity separately as we have done. 

An approach to this problem which we have not pursued here would be to design an algorithm or an approximation algorithm that returns the size of the minimum connected dominating set of an $n \times m$ grid in time polynomial in $n$ and $m$. It is important to note that an approximation algorithm would lead to an \emph{upper bound} on the connected domination number of the grid, while our work has focused on lower bounds. The constructions described in section \ref{const} for example, would lead to a trivial approximation algorithm. The gap between our current lower and upper bounds is linear in $m$ and $n$, which means that this would be asymptotically better than any $(1+\epsilon)$-approximation algorithm (the input size is $O(mn)$). A non-constructive approach to obtaining upper bounds, for example using the probabilistic method could also be tried. 

Our approach has been completely analytical. In \cite{griddomset1} the authors used a computational approach to answer the analogous questions about the domination number of grids, using dynamic programming algorithms to calculate the minimum value of a similar loss function near the boundary of the grid. The connected dominating set problem is one of a more `global' nature than the dominating set problem, as it involves connectivity as a constraint. This entails a very different set of challenges and a computational approach to the problem would likely require new techniques.

%
%
%
\bibliographystyle{splncs04}
\bibliography{gridgraph}

\begin{thebibliography}{10}
\providecommand{\url}[1]{\texttt{#1}}
\providecommand{\urlprefix}{URL }
\providecommand{\doi}[1]{https://doi.org/#1}

\bibitem{chang}
Chang, T.: Domination Numbers of Grid Graphs. Ph{D} {T}hesis, University of
  South Florida (1992)

\bibitem{CLARK1990165}
Clark, B.N., Colbourn, C.J., Johnson, D.S.: Unit disk graphs. Discrete
  Mathematics  \textbf{86}(1),  165 -- 177 (1990).
  \doi{10.1016/0012-365X(90)90358-O}

\bibitem{du2012polynomial}
Du, H., Ye, Q., Zhong, J., Wang, Y., Lee, W., Park, H.: Polynomial-time
  approximation scheme for minimum connected dominating set under routing cost
  constraint in wireless sensor networks. Theoretical Computer Science
  \textbf{447},  38--43 (2012). \doi{10.1016/j.tcs.2011.10.010}

\bibitem{FUJIE20031931}
Fujie, T.: An exact algorithm for the maximum leaf spanning tree problem.
  Computers \& Operations Research  \textbf{30}(13),  1931 -- 1944 (2003).
  \doi{10.1016/S0305-0548(02)00117-X}

\bibitem{MLSTApprox}
Galbiati, G., Maffioli, F., Morzenti, A.: A short note on the approximability
  of the maximum leaves spanning tree problem. Information Processing Letters
  \textbf{52}(1),  45 -- 49 (1994). \doi{10.1016/0020-0190(94)90139-2}

\bibitem{garey}
Garey, M.R., Johnson, D.S.: Computers and Intractability: A Guide to the Theory
  of NP-Completeness. W. H. Freeman \& Co., USA (1979)

\bibitem{griddomset1}
Gonçalves, D., Pinlou, A., Rao, M., Thomassé, S.: The domination number of
  grids. SIAM Journal on Discrete Mathematics  \textbf{25}(3),  1443--1453
  (2011). \doi{10.1137/11082574}

\bibitem{hunt1998nc}
Hunt~III, H.B., Marathe, M.V., Radhakrishnan, V., Ravi, S.S., Rosenkrantz,
  D.J., Stearns, R.E.: {N}{C}-approximation schemes for {NP}-and {PSPACE}-hard
  problems for geometric graphs. Journal of algorithms  \textbf{26}(2),
  238--274 (1998). \doi{10.1006/jagm.1997.0903}

\bibitem{bipartiteNP}
Li, B.P., Toulouse, M.: Variations of the maximum leaf spanning tree problem
  for bipartite graphs. Information Processing Letters  \textbf{97}(4),  129 --
  132 (2006). \doi{10.1016/j.ipl.2005.10.011}

\bibitem{lichtenstein1982planar}
Lichtenstein, D.: Planar formulae and their uses. SIAM journal on computing
  \textbf{11}(2),  329--343 (1982). \doi{10.1137/0211025}

\bibitem{gridbounds}
Lie, P., Toulouse, M.: Maximum leaf spanning tree problem for grid graphs.
  JCMCC. The Journal of Combinatorial Mathematics and Combinatorial Computing
  \textbf{73} (2010),
  \url{http://www.cs.umanitoba.ca/~lipakc/gridgraph-aug6-08.pdf}

\bibitem{lu1998approximating}
Lu, H.I., Ravi, R.: Approximating maximum leaf spanning trees in almost linear
  time. Journal of algorithms  \textbf{29}(1),  132--141 (1998).
  \doi{10.1006/jagm.1998.0944}

\bibitem{solis20172}
Solis-Oba, R., Bonsma, P., Lowski, S.: A 2-approximation algorithm for finding
  a spanning tree with maximum number of leaves. Algorithmica  \textbf{77}(2),
  374--388 (2017). \doi{10.1007/s00453-015-0080-0}

\bibitem{west}
West, D.B.: Introduction to Graph Theory. Prentice Hall, 2 edn. (September
  2000)

\bibitem{uses1}
Wu, J., Li, H.: On calculating connected dominating set for efficient routing
  in ad hoc wireless networks. In: Proceedings of the 3rd International
  Workshop on Discrete Algorithms and Methods for Mobile Computing and
  Communications. p. 7–14. DIALM '99, Association for Computing Machinery,
  New York, NY, USA (1999). \doi{10.1145/313239.313261}

\end{thebibliography}
\end{document}